\documentclass[sn-mathphys-num]{sn-jnl}

\usepackage{graphicx}%
\usepackage{multirow}%
\usepackage{amsmath,amssymb,amsfonts}%
\usepackage{amsthm}%
\usepackage{mathrsfs}%
\usepackage[title]{appendix}%
\usepackage{xcolor}%
\usepackage{textcomp}%
\usepackage{manyfoot}%
\usepackage{booktabs}%
\usepackage{algorithm}%
\usepackage{algorithmicx}%
\usepackage{algpseudocode}%
\usepackage{listings}%
\usepackage[normalem]{ulem}

\theoremstyle{thmstyleone}%
\newtheorem{theorem}{Theorem}
\theoremstyle{thmstyletwo}%

\newtheorem{lemma}[theorem]{Lemma}

\newtheorem{rmk}[theorem]{Remark}

\theoremstyle{thmstylethree}%

\newcommand{\R}{\mathbb{R}}
\newcommand{\HH}{\mathbb{H}}
\newcommand{\QLM}{QLM}
\newcommand{\QLE}{QLE}

\newcommand{\pst}{\mathcal{M}^{3,1}}

\begin{document}

\title[Article Title]{Strong field behavior of Wang-Yau quasi-local energy}

\author*[1]{\fnm{Bowen} \sur{Zhao}}\email{bowenzhao@bimsa.cn}

\author[1]{\fnm{Lars} \sur{Andersson}}

\author[2]{\fnm{Shing-Tung} \sur{Yau}}

\affil*[1]{\orgname{Beijing Institute of Mathematical Sciences and Applications}, \orgaddress{\city{Beijing}, \postcode{101408}, \country{China}}}

\affil[2]{\orgdiv{Yau Mathematical Sciences Center}, \orgname{Tsinghua University},\city{Beijing}, \postcode{100084}, \country{China}}



\abstract{We look at the strong field behavior of the Wang-Yau quasi-local energy. In particular, we examine the limit of the Wang-Yau quasi-local energy as the defining spacelike $2$-surface $\Sigma$ approaches an apparent horizon from outside. Assuming that coordinate functions of the isometric embedding are bounded in $W^{2,1}$ and mean curvature vector of the image surface remains spacelike, we find that the limit falls in two exclusive cases: 1) If the horizon \textit{cannot} be isometrically embedded into $R^3$, the Wang-Yau quasi-local energy blows up as $\Sigma$ approaches the horizon while the optimal embedding equation is not solvable for $\Sigma$ near the horizon; 2) If the horizon \textit{can} be isometrically embedded into $R^3$, the optimal embedding equation is solvable up to the horizon with the unique solution at the horizon corresponding to isometric embedding into $\R^3$ and the Wang-Yau quasi-local mass admits a finite limit at the horizon. We discuss the implications of our results in the conclusion section.}




\maketitle

\section{Introduction}\label{sec:intro}

The energy-momentum tensor in general relativity gives a point-wise energy density for matter only. The equivalence principle of Einstein gravity theory excludes the existence of a local gravitational energy density.
On the other hand, there exist well-defined notions of total gravitational energy for a vaccum, asymptotically flat spacetime, e.g. ADM mass at spatial infinity and Bondi mass at null infinity. A quasi-local mass aims to measure gravitational energy within a finite, extended region, whose definition remains a major unresolved problem in general relativity \cite{Penrose1982unsolved,coley2017open}.

Since there exists no local gravitational energy density, the gravitational energy of a finite region cannot be in the form of a volume integral, as in the case of matter energy. Rather, most definitions of quasi-local mass are in the form of surface integral over a space-like $2$-surface $\Sigma$. The existing proposals for quasi-local energy/mass, e.g. Hawing mass \cite{hawking1968hawkingmass}, Bartnik mass \cite{bartnik1989newqlm}, Brown-York mass \cite{brown1993quasilocal} and Liu-Yau mass \cite{liuYau2003,liuYau2006} all possess some undesirable features\cite{Szabados:2009review,wang2015four}. 
For example, the Hawking mass could be negative as well as positive even in the Minkowski spacetime while the Brown-York mass depends on the choice of a spacelike $3$-manifold bounded by the $2$-surface $\Sigma$ under consideration. Moreover, there exist surfaces in the Minkowski spacetime with strictly positive Liu--Yau mass as well as Brown--York mass \cite{murchadha2004comment}, in contrast to expectations that surfaces in the Minkowski spacetime should have zero mass. In order to overcome these difficulties,  Wang and Yau proposed a definition that incorporates both metric and momentum information about the $2$-surface under consideration \cite{wangyau2009cmp,wang2009quasilocalPRL}. Their definition is intrinsic to the spacelike $2$-surface and most importantly, the positivity was proved assuming some technical conditions.

The Wang-Yau quasi-local mass was shown to be consistent with the ADM mass at spatial infinity \cite{WangYau:2010spatialinf} and Bondi mass at null infinity \cite{Chen:2010tz}. It also attains a small-sphere limit that is related to either the energy-momentum tensor (non-vacumm spacetime) or the Bel-Robinson tensor (vaccum spacetime) \cite{ChenWangYau:2018smalllimit}. In establishing these limits, the fourth order optimal embedding equation was solved perturbatively, assuming the solution is close to an isometric embedding into $R^3$.
These limits all fall in the weak field regime. It is also important to examine the limit of Wang-Yau quasi-local energy in the strong field regime. In this paper, we examine the limit of the Wang-Yau quasi-local energy as the defining spacelike 2-surface $\Sigma$ approaches an apparent horizon from outside.

Recall that a marginally outer trapped surface (MOTS) is a space-like $2$-surface whose outgoing null expansion $\Theta_+$ vanishes. Here by apparent horizon we mean the outermost MOST. If the apparent horizon is strictly stable in the sense of MOTS, one can build a horizon-crossing foliation $\Sigma(s), s\in(-\epsilon,\epsilon)$ with $\Sigma(0)=\Sigma$ (see Lemma\ \ref{lemma:foliation} below). For each leaf $\Sigma(s)$ of the foliation, consider an isometric embedding $X_s: \Sigma(s) \to \R^{3,1}$.
We prove the following theorem assuming that isometric embeddings $X_s's$ are $C^2$.
\\
\begin{theorem}\label{thm:main_C2}
    Let $\Sigma$ be a strictly stable apparent horizon whose in-going null expansion $\Theta_-$ is negative and bounded away from $0$. For each leaf $\Sigma(s)$ of a cross-horizon foliation, assume that the isometric embedding $X_s$ is $C^2$ and the mean curvature vector of the isometric embedding image $X_s(\Sigma_s)$ is space-like. Then the following two exclusive cases apply.

   1) If the apparent horizon $\Sigma$ cannot be isometrically embedded into $\R^3$, the Wang-Yau quasi-local energy blows up to $\infty$ at the horizon while the optimal embedding equation does not admit a $C^2$ solution near the horizon.
    
   2) If the apparent horizon $\Sigma$ can be isometrically embedded into $\R^3$, the optimal embedding equation is solvable up to $\Sigma$ with the unique solution at the horizon corresponding to the isometric embedding into $R^3$. The Wang-Yau quasi-local mass admits a finite limit at the horizon
   $$\frac{1}{8\pi}\int_\Sigma |H_0|$$
   
\end{theorem}

\begin{rmk}\label{rmk:chen_W21}
In the construction of the Wang-Yau quasi-local energy, one needs to fix a future-pointing, unit time-like vector $T_0$ in $\R^{3,1}$.
Let $\pi:\R^{3,1}\to \R^3$ be the projection along $T_0$ and $\hat{k}$ be the mean curvature of $\pi\circ X_s (\Sigma(s))$ in $\R^3$. If one further assumes that $\hat{k}\in L^1$, one can relax the $C^2$ condition on isometric embedding $X_s$ used in Theorem \ref{thm:main_C2} to $X_s$ in  $W^{2,1}$ \cite{ChenPN}. This weaker assumption allows coordinate functions of the isometric embedding $X_s$ to blow up but the same conclusions in Theorem \ref{thm:main_C2} hold. We note that a different method to prove the second case was shown in an independent study \cite{chen2024near}.
\end{rmk}


The rest of his paper is organized as follows. We briefly review the setup of Wang-Yau quasi-local energy/mass in section \ref{sec:review}. We examine the limit of Wang-Yau quasi-local energy and solution to the optimal embedding equation as $\Sigma$ approaches the horizon from outside in section \ref{sec:limit_horizon}. We discuss the implication of our results about the Wang-Yau quasi-local energy and the optimal embedding equation in section \ref{sec:concludingrmk}.

\section{Review of Wang-Yau quasi-local energy}\label{sec:review}
\subsection{Definition of Wang-Yau quasi-local energy}
The Wang-Yau quasi-local energy is defined as a difference between a reference term and a physical term, like Brown-York mass and Liu-Yau mass. This ensures that for surfaces lying in the reference spacetime, typically taken as $\R^{3,1}$, the quasi-local energy vanishes.
Let $\Sigma$ be a spacelike $2$-surface from a general spacetime $\pst$. We denote covariant derivative on $\Sigma$ by $\nabla$ and covariant derivative in $\pst$ by $\overline{\nabla}$. The Wang-Yau quasi-local energy requires the following information about $\Sigma\subset \pst$: the metric $d\sigma^2$, the mean curvature vector $H$ which is assumed to be spacelike and the connection one-form
\begin{align*}
\alpha_H = \langle \overline{\nabla} \frac{-H}{|H|}, \frac{J}{|H|}\rangle
\end{align*}
where $|H|=|J|=\sqrt{\langle H,H\rangle}$ with $J$ the reflection of $H$ along the ingoing null cone in the normal bundle $N\Sigma$. 

Let $X:\Sigma \to \R^{3,1}$ be an isometric embedding and denote the image by $\Sigma_0=X(\Sigma)$. Also fix a future-pointing, timelike, unit vector $T_0$ in $\R^{3,1}$. Let $\pi:\R^{3,1}\to \R^3$ be the projection along $T_0$.
Denote the projected image by $\hat{\Sigma}=\pi\circ X (\Sigma)$.
One records the following information about the isometric embedding $\Sigma_0$, the time function along $T_0$
$$\tau = -\langle X, T_0\rangle$$ 
the mean curvature vector $H_0$ and the connection one-form 
\begin{align*}
\alpha_{H_0} = \langle \nabla^{\R^{3,1}} \frac{-H_0}{|H_0|}, \frac{J_0}{|H_0|}\rangle
\end{align*}
where $\nabla^{\R^{3,1}}$ is the covariant derivative for $\R^{3,1}$ and $J_0$ the reflection of $H_0$ along the ingoing null cone in the normal bundle $N\Sigma_0$. 

The isometric embedding $X$ induces a natural bijection between the two tangent spaces $T\Sigma$ and $T\Sigma_0$. One builds a bijection between the two normal bundles $N\Sigma$ and $N\Sigma_0$ through the canonical gauge condition
\begin{equation}\label{eq:canonical_gauge}
 \langle T_0, H_0\rangle = \langle T,H\rangle   
\end{equation}
The unit, future-pointing, time-like vector $T_0$ ($T$) determines a unique timelike unit normal $u_0$ ($u$) for $\Sigma_0$ ($\Sigma$) through
 $$T_0 = \sqrt{1+|\nabla\tau|^2} u_0 -\nabla \tau,\quad T = \sqrt{1+|\nabla\tau|^2} u -\nabla \tau$$
Complete the basis for the normal bundle through
$$\langle v_0, u_0\rangle =\langle v_0, T_0 \rangle=\langle v, u\rangle =\langle v, T \rangle=0$$
This yields a \textit{canonical} basis $\{v_0,u_0\}$ for the normal bundle $N\Sigma_0$ and  a \textit{canonical} basis $\{v, u\}$ for the normal bundle $N\Sigma$.
The bijection between $N\Sigma$ and $N\Sigma_0$ is built through $u\longleftrightarrow u_0$ and $v\longleftrightarrow v_0$.

Further, \cite{wang2009isometric} employed the frame independent mean curvature vector $H_0$ and its conjugate vector $J_0$ to define a \textit{reference} basis for $N\Sigma_0$, $$e_{H_0} = -\frac{H_0}{|H_0|}, e_{J_0} = \frac{J_0}{|H_0|}$$
The basis transformation between the \textit{reference} and the \textit{canonical} basis is
\begin{align*}
    v_0 &= \cosh\theta_0 \, e_{H_0} - \sinh\theta_0 \, e_{J_0}\\
    u_0 &= -\sinh\theta_0 \, e_{H_0} + \cosh\theta_0 \, e_{J_0}
\end{align*}
with $$\sinh \theta_0 = \frac{-\Delta \tau}{|H_0|\sqrt{1+|\nabla \tau|^2}}$$
Similarly, oen can define $\{e_H = -\frac{H}{|H|}, e_J=\frac{J}{|J|}\}$  as the \textit{reference} basis for $N\Sigma$. 
The basis transformation in $\pst$ is similar as above but
with 
\begin{equation}\label{eq:angle_spacelike}
    \sinh \theta = \frac{-\Delta\tau}{|H|\sqrt{1+|\nabla \tau|^2}}
\end{equation}

Given $\tau$ or given a pair $(X, T_0)$, the Wang-Yau quasi-local energy (QLE) is defined as
\begin{align}
    QLE 
    &= \frac{1}{8\pi} \bigg( \int_{\hat{\Sigma}} \hat{k} - \int_\Sigma \sqrt{1+|\nabla\tau|^2}|H|\cosh\theta - \nabla\tau\cdot \nabla\theta -\alpha_{H}(\nabla\tau)\bigg) \label{eq:QLE_reference_basis}
\end{align}
where $\hat{k}$ is the mean curvature of $\hat{\Sigma}=\pi\circ X(\Sigma)$ in $\R^3$. The first integral is the reference term while the second term is the physical term.
We also recall that 
$$\int_{\hat{\Sigma}} \hat{k} = \int_\Sigma \sqrt{(\Delta\tau)^2+|H_0|^2(1+|\nabla\tau|^2)} - \nabla \tau \cdot \nabla\theta_0  -\alpha_{H_0}(\nabla \tau)$$

\subsection{The optimal embedding equation}
By a theorem of Pogorelov \cite{pogorelov1964embedding}, a positive definite metric on a topological $2$-sphere $\Sigma$ with Gaussian curvature $K_\Sigma >-\kappa^2$ can be isometrically embedded into a hyperboloid of $\R^{3,1}$, $\{X^\mu\in \R^{3,1}: \eta_{\mu\nu} X^\mu X^\nu = - \frac{1}{\kappa^2}\}$. So an isometric embedding $X:\Sigma \to \R^{3,1}$ is far from unique. Wang and Yau defined the minimum of quasi-local energy among all ``admissible'' 
\footnote{The admissibility condition is required in the positivity proof of Wang-Yau quasi-local mass \cite{wang2009isometric}. It consists of two conditions: 1) The Gussian curvature of $d\hat{\sigma}^2=d\sigma^2+d\tau^2$ is positive; 2) The generalized mean curvature $-\sqrt{1+|\nabla\tau|^2}\langle H,e_3'\rangle-\alpha_{e_3'}(\nabla\tau) > 0$ for some spacelike unit normal $e_3'$ determined by Jang's equation. Whether one can prove the positivity or boundedness without assuming admissbility conditions remains open. In this study, we are not concerned with the positivity so we ignore such admissbility conditions. See \cite{zhaoAY2024remarks} for a detailed dicussion.} 
isometric embeddings as the quasi-local mass. The variational problem can be solved through varying the time function $\tau=\langle X, T_0 \rangle$. The Euler-Lagrangian equation or the optimal embedding equation is 
\begin{align}
    \nabla \cdot j &=0, \\
    j &= \rho\nabla\tau - \alpha_{v_0} + \alpha_v=\rho\nabla\tau - \alpha_{H_0} -\nabla\theta_0 + \alpha_H + \nabla\theta \nonumber
\end{align}
where $\alpha_{v_0} = \langle \nabla^{\R^{3,1}} v_0, u_0\rangle = \alpha_{H_0} + \nabla \theta_0$ and $\alpha_{v} = \langle \overline{\nabla} v, u\rangle = \alpha_{H} + \nabla \theta$. The solvability of the optimal embedding equation is only proved in the large- and small-sphere limits \cite{Chen:2010tz,chen2018small}. The critical point is shown to locally minimize the quasi-local energy if the mean curvature condition $|H_0|>|H|>0$ is satisfied \cite{ChenWangYau:2014cmp}.

\subsection{Mass surface density and energy current}\label{sec:review_rho_j}
Lastly, we note that the canonical gauge \eqref{eq:canonical_gauge} is equivalent to 
\begin{align}\label{eq:canonical_gauge_u}
    \langle u_0,H_0 \rangle = \langle u, H \rangle
\end{align}
which matches the ``time'' component of mean curvature vectors. The difference in the other component of mean curvature vectors measures the gravitational energy through the mass energy surface density
$$\rho = \frac{\langle H,v\rangle- \langle H_0,v_0\rangle}{\sqrt{1+|\nabla \tau|^2}}=\frac{\sqrt{|H_0|^2+\frac{(\Delta\tau)^2}{1+|\nabla\tau|^2}}-\sqrt{|H|^2+\frac{(\Delta\tau)^2}{1+|\nabla\tau|^2}}}{\sqrt{1+|\nabla\tau|^2}}$$
In terms of $\rho$ and $j$, the quasi-local energy is
$$\QLE = \frac{1}{8\pi}\int_\Sigma \rho + j\cdot \nabla\tau$$ 
which reduces to the quasi-local mass after imposing optimal embedding equation $\nabla\cdot j=0$
$$\QLM = \frac{1}{8\pi}\int_\Sigma \rho$$

\section{Limit of Wang-Yau quasi-local energy at an apparent horizon}\label{sec:limit_horizon}
We consider the situation that the defining $2$-surface $\Sigma$ approaches an apparent horizon along a foliation. In speaking about the limit of $|H|\to 0$ below, we also implicitly refer to such a foliation. So let us start with building an explicit foliation near an horizon. 
We only consider strictly stable horizons here and leave the marginally stable case to a separate paper.
\begin{lemma}\label{lemma:foliation}
Let $\Sigma$ be a strictly stable horizon, whose ingoing null expansion $\Theta_-<0$ is bounded away from $0$.
    Around $\Sigma$, one can build a horizon crossing foliation $\Sigma(s)$ for $s\in (-\epsilon,\epsilon)$, with $\Sigma =\Sigma(s=0)$ and positive $s$ labelling surfaces lying outside $\Sigma$ while negative $s$ labelling surfaces lying inside $\Sigma$. For this foliation, $\frac{\nabla|H|}{|H|}$, $\alpha_H$ and $\nabla\cdot \alpha_H$ all remain bounded as $|H|\to 0$.
\end{lemma}
\begin{proof}
    Let us consider the outside region of the horizon, i.e. the untrapped region. A similar argument applies to the trapped region by reversing the direction of the normal. First recall that an apparent horizon $\Sigma$ as the outermost MOTS is stable \cite{andersson2011jangreview}. Let $L_\Sigma$ be the MOST stability operator for $\Sigma$, which then has nonnegative principal eigenvalue $\lambda\geq 0$. Since we assume $\Sigma$ is strictly stable, $\lambda>0$. Let $\varphi$ be the corresponding positive eigenfunction, i.e.
    $L_\Sigma \, \varphi = \lambda \,\varphi > 0$.
    Consider the flow
    $$\frac{d F}{ds} = \varphi \, \nu$$
    with $F(s=0)=\Sigma$ being the horizon. Recall that the variation of outgoing null expansion at $\Sigma$ under this flow is
    $$\frac{d\Theta_+}{d s}|_{s=0}=L_\Sigma \, \varphi =\lambda \varphi > 0$$ 
    For small $s$, as $\Sigma_s$ approaches $\Sigma$, i.e. $s \to 0^+$,
    $$\frac{{\nabla}\Theta_+}{\Theta_+} \approx  \frac{ {\nabla} (\lambda\varphi s)}{\lambda \varphi s} = \lambda\frac{{\nabla} \varphi}{\varphi} <\infty$$
    On the other hand,
    $$\frac{{\nabla} |H|}{|H|} = \frac{{\nabla} \sqrt{-\Theta_+\Theta_-}}{-\Theta_+\Theta_-} = \frac{1}{2}\frac{{\nabla} \Theta_+}{\Theta_+}+\frac{1}{2}\frac{{\nabla}|\Theta_-|}{|\Theta_-|} <\infty $$
    where we used the assumption $|\Theta_-|=-\Theta_-$ is bounded away from $0$.
    Recall that 
    $$\alpha_H= \frac{{\nabla} |H|}{|H|} - \frac{{\nabla} \Theta_-}{\Theta_-}-\frac{1}{2}\langle \overline{\nabla}\, l^-,l^+\rangle$$
    Then it follows directly that $\alpha_H$ and $\nabla\cdot \alpha_H$ also remain bounded. 
\end{proof}
\begin{rmk}
    One can actually build a cross-horizon foliation by constant expansion surfaces as shown in \cite{zhaoKW2021blowupJang}.
\end{rmk}

\subsection{Canonical gauge condition at the horizon}
We first observe that the canonical gauge condition \eqref{eq:canonical_gauge} cannot be extended to an horizon where the mean curvature vector $H$ becomes null.
Consider a future-untrapped $2$-surface $\Sigma$ with outgoing null expansion $\Theta_+>0$ and ingoing null expansion $\Theta_-<0$. A timelike, future-pointing, unit vector $T$ can be decomposed as
$$T = N u + \Vec{N}$$ 
where $N>0$, $\Vec{N}\in T\Sigma$ and $u$ is a future-pointing, unit normal vector of $\Sigma$. 
Let $v$ completes $u$ to be an orthonormal basis for $N\Sigma$. In this basis, one has
$$H=-kv +pu, \quad \Theta_\pm=p \pm k$$ 
where $k=-\langle H, v\rangle=\frac{\Theta_+-\Theta_-}{2}>0$ and $p=-\langle H, u\rangle$. As $\Sigma$ approaches an apparent horizon, $\Theta_+\to 0$ and $|H|=\sqrt{-\Theta_+\Theta_-} \to 0$ while $k=-p>0$ at the horizon.
Then at the horizon,
$$\langle H, T\rangle=\langle -k(u+v),  N u + \Vec{N}\rangle = kN >0$$
On the other hand, $\langle T_0,H_0\rangle$ is either identically zero or changes sign over $\Sigma$ as
    $$\langle T_0,H_0\rangle  = -\Delta \tau$$
where $H_0=\Delta X$ and $\tau=-\langle X, T_0\rangle $ is the time function for an isometric embedding $X:\Sigma \to \R^{3,1}$. Therefore, the canonical gauge \eqref{eq:canonical_gauge} cannot hold at an horizon.

Moreover, the frame-independent \textit{reference} basis $\{e_H=\frac{-H}{|H|}, e_J=\frac{J}{|H|}\}$ for the normal bundle $N\Sigma$ fails to be defined since 
$H=-J$ and $|H|= 0$ at an horizon. Nevertheless, one can examine the limit of Wang-Yau quasi-local energy expression as $|H|\to 0^+$. 

\subsection{The quasi-local energy and the optimal embedding in the limit $|H|\to 0^+$}
Let us start with rewriting the Wang-Yau quasi-local energy and the optimal embedding equation. Recall that 
\begin{equation}
    \sinh^{-1}(x) = \ln(x+\sqrt{x^2+1}), \quad \forall x\in \R \label{eq:sinh}
\end{equation}
and that for $\Sigma$ with spacelike $H$, the connection $1$-form $\alpha_H $ on $\Sigma$ takes the form
\begin{align}
    \alpha_H &= \langle \overline{\nabla}\, \frac{-H}{|H|}, \frac{J}{|H|}\rangle \nonumber \\
    &=  \frac{-1}{|H|^2}\langle \overline{\nabla}\, \frac{\Theta_+ l^- + \Theta_- l^+}{2}, \frac{\Theta_+ l^- - \Theta_- l^- }{2}\rangle \nonumber \\
    &= \frac{{\nabla} |H|}{|H|} - \frac{{\nabla} \Theta_-}{\Theta_-}-\frac{1}{2}\langle \overline{\nabla}\, l^-,l^+\rangle \label{eq:alpha_H_rewrite}
\end{align}
where $l^{\pm}=u\pm v$ are normalized to satisfy $\langle l^+, l^-\rangle = -2$.
Then using \eqref{eq:sinh} and \eqref{eq:alpha_H_rewrite}, one has
\begin{align*}
\alpha_H + \nabla\theta &=   \nabla \ln|H|- \nabla\ln|\Theta_-|\,-\frac{1}{2}\langle \overline{\nabla}\, l^-,l^+\rangle + \nabla \sinh^{-1} \frac{-\Delta\tau}{|H|\sqrt{1+|\nabla\tau|^2}}\\
&= \nabla\ln \bigg[\frac{-\Delta\tau}{\sqrt{1+|\nabla\tau|^2}} + \sqrt{|H|^2+\frac{(\Delta\tau)^2}{1+|\nabla\tau|^2}}\bigg]
- \nabla\ln|\Theta_-| \, -\frac{1}{2}\langle \overline{\nabla}\, l^-,l^+\rangle
\end{align*}
From now on we denote $$y=\ln \bigg[\frac{-\Delta\tau}{\sqrt{1+|\nabla\tau|^2}} + \sqrt{|H|^2+\frac{(\Delta\tau)^2}{1+|\nabla\tau|^2}}\bigg]$$
We arrive at the following form of the Wang-Yau quasi-local energy (without imposing the optimal embedding equation) and the optimal embedding equation.
\begin{lemma}
The Wang-Yau quasi-local energy can be written as 
\begin{align}
    & 8\pi \, QLE = \nonumber \\
    & \int_{\hat{\Sigma}} \hat{k} - \int_\Sigma \sqrt{(\Delta\tau)^2 + |H|^2(1+|\nabla\tau|^2)} + \nabla\tau \cdot \nabla\ln|\Theta_-| -\frac{1}{2}\langle \overline{\nabla}_{\nabla\tau}\, l^-,l^+\rangle +\tau\Delta y \label{eq:QLE_trouble_term}
\end{align}
while the optimal embedding equation can be written as 
\begin{align}
    \Delta y  = -\nabla\cdot\bigg[\rho \nabla\tau -\nabla\theta_0-\alpha_{H_0} - \nabla\ln|\Theta_-| +\frac{1}{2}\langle \overline{\nabla}\, l^-,l^+\rangle \bigg]  \label{eq:OEE_trouble_term}
\end{align}
\end{lemma}

To study the limit of $|H|\to 0^+$, we begin with the following observation.
\begin{lemma}\label{lemma:bounded}
    Let $\Sigma$ be an apparent horizon whose in-going null expansion $\Theta_-$ is negative and bounded away from $0$. Consider $C^2$ isometric embedding $X: \Sigma(s) \to \R^{3,1}$. Assume that the mean curvature vector of $X(\Sigma(s))$ is space-like. Then in taking the limit $|H|\to 0^+$, the term involving $\Delta y$ is the only term that may blow up in both the quasi-local energy expression and the optimal embedding equation. 
\end{lemma}
\begin{proof}

The reference term of the quasi-local energy
$$\int_{\hat{\Sigma}} \hat{k} = \int_\Sigma \sqrt{(\Delta\tau)^2+|H_0|^2(1+|\nabla\tau|^2)} + \Delta \tau\,\sinh^{-1}\frac{-\Delta\tau}{|H_0|\sqrt{1+|\nabla\tau|^2}}  -\alpha_{H_0}(\nabla \tau)$$
clearly remains bounded as $|H|\to 0^+$, assuming $|H_0|>0$ and $\tau\in C^2$.

Among the physical terms, 
$$\sqrt{(\Delta\tau)^2+|H|^2(1+|\nabla\tau|^2)}\to |\Delta\tau|, \quad |H|\to 0^+$$ 
clearly has a bounded limit.
Since $|\Theta_-|=-\Theta_-$ is assumed to be bounded away from $0$ for the horizon under consideration, then the following two terms also have bounded limits
$$\nabla\tau \cdot \nabla\ln|\Theta_-| +\frac{1}{2}\langle \overline{\nabla}_{\nabla\tau}\, l^-,l^+\rangle$$
Therefore, under stated assumptions, all terms except $\int_\Sigma \tau\Delta y$ in the quasi-local energy expression \eqref{eq:QLE_trouble_term} remain bounded as $|H|\to 0^+$.

Now examine the optimal embedding equation. It is easy to see that
$$\rho\nabla\tau = \frac{\nabla\tau}{\sqrt{1+|\nabla\tau|^2}}(\sqrt{|H_0|^2+\frac{(\Delta\tau)^2}{1+|\nabla\tau|^2}}-\sqrt{|H|^2+\frac{(\Delta\tau)^2}{1+|\nabla\tau|^2}}) $$
remains bounded as $|H|\to 0^+$. The other terms except $\Delta y$ does not involve $|H|$ so all remain bounded in taking the limit $|H|\to 0^+$, under stated assumptions.
\end{proof}

Therefore the problem reduces to study the integral
\begin{align}\label{eq:integral_blowup}
    \mathcal{I}= \int_\Sigma \tau \Delta y &=\int_\Sigma \Delta \tau \ln \bigg[\sqrt{|H|^2+\frac{(\Delta\tau)^2}{1+|\nabla\tau|^2}} -\frac{\Delta\tau}{\sqrt{1+|\nabla\tau|^2}}\bigg]
\end{align}
\begin{lemma}\label{lemma:trouble_term}
    Assume the same conditions as in Lemma \ref{lemma:bounded}. 
    Then the integral $\mathcal{I}$ \eqref{eq:integral_blowup}
    blows up unless $\tau$ approaches to a constant function as $|H|\to 0^+$.
\end{lemma}
\begin{proof}

Let $z:=\frac{\Delta \tau}{\sqrt{1+|\nabla\tau|^2}}$ and take $|z| < C$ on the compact surface $\Sigma$.
We need to examine the limit of the function $$h(z)=z\ln(\sqrt{z^2+a^2}-z), \quad \text{as } \, a\to 0$$ 
For $z\leq 0$, as $a\to 0$
    $$h(z)= z\ln (\sqrt{z^2+a^2}+|z|) \to z\ln |2z|  \leq C\ln 2 + \max\{C\ln(C),\frac{1}{e} \} $$
    For $0<z\leq a$, let $\lambda=\frac{z}{a}$ then $0<\lambda\leq 1$. As $a\to 0$
    $$h(z) = \lambda a \ln\big[a\sqrt{1+\lambda^2}-\lambda a\big] = \lambda a\ln a + \lambda a \ln(\sqrt{1+\lambda^2}-\lambda) \to 0$$
    where we used that $\lim_{a \to 0} a\ln|a| =0$.
    
    For $z>a$, let $\lambda=\frac{a}{z}>0$, which can be arbitrarily small as $a\to 0$. Now
    $$h(z) = z\ln z + z\ln(\sqrt{1+\lambda^2}-1) \to z\ln z + z\ln (\frac{\lambda^2}{2}-\frac{\lambda^4}{8}+...)$$
    whose modulus can be arbitrary large as $|h(z)|\geq  z \ln \frac{2}{\lambda^2} - C|\ln C|$ for $\lambda \ll 1$.
    Therefore, if $\tau$ does not approach to a constant function as $|H|\to 0^+$, 
    $$\mathcal{I}= \int_\Sigma \sqrt{1+|\nabla\tau|^2} \, h(z)$$ can be arbitrarily negative.
\end{proof}

From \eqref{eq:QLE_trouble_term} and Lemma \ref{lemma:trouble_term}, it directly follows that that $\QLE \to +\infty$ as $|H|\to 0^+$, unless $\tau$ approaches to a constant function toward the horizon. We can now prove the $C^2$ version of our main result, i.e. Theorem \ref{thm:main_C2}. 
\begin{proof}
If the horizon cannot be isometrically embedded into $\R^3$, $\tau$ cannot take a constant function limit and hence the Wang-Yau quasi-local energy blows up from Lemma \ref{lemma:trouble_term}. Further note that 
\begin{align}\label{eq:OEE_solvability}
    \int_\Sigma \tau \Delta y = \int_\Sigma \nabla\tau \cdot\bigg[\rho \nabla\tau -\nabla\theta_0-\alpha_{H_0} - \nabla\ln|\Theta_-| -\frac{1}{2}\langle \overline{\nabla}\, l^-,l^+\rangle \bigg]
\end{align}
is a solvability condition for the optimal embedding equation \eqref{eq:OEE_trouble_term}. The RHS remains bounded as $|H|\to 0^+$, as was argued in Lemma \ref{lemma:bounded}. However, the LHS could be arbitrarily large as $|H|\to 0^+$ if the horizon cannot be embedded into $\R^3$. Therefore, the optimal embedding equation does not admit a $C^2$ solution near the horizon.

On the other hand, if the horizon can be isometrically embedded into $\R^3$, a formal series solution to the optimal embedding equation was constructed in \cite{Chenthesis}
$$\tau = s \, \tau^{(1)}+\cdots$$
where $s$ is the foliation parameter as in Lemma \ref{lemma:foliation}. We argue here that constant $\tau$ is the unique solution to the optimal embedding equation in the limit of $|H|\to 0^+$.
If the optimal embedding equation is solvable in $C^2$, as reviewed in section \ref{sec:review_rho_j}, the Wang-Yau quasi-local mass is simply
$$QLM = \frac{1}{8\pi}\int_\Sigma \rho=\frac{1}{8\pi}\int_\Sigma \frac{\sqrt{|H_0|^2+\frac{\Delta\tau}{1+|\nabla\tau|^2}}-\sqrt{|H|^2+\frac{\Delta\tau}{1+|\nabla\tau|^2}}}{\sqrt{1+|\nabla\tau|^2}}$$ 
which clearly admits a well-defined limit as $|H|\to 0^+$. Assume the optimal embedding equation admits a non-constant solution at the horizon. This will yield a blowup quasi-local mass and hence a contradiction. Therefore, the constant $\tau$ solution is the unique solution to the optimal embedding equation in the limit $|H|\to 0^+$. 
\end{proof}



\begin{rmk}
As pointed out by Po-Ning Chen \cite{ChenPN},  if we further assume that the mean curvature $\hat{k}$ of $\hat{\Sigma}=\pi\circ X(\Sigma)$ in $\R^3$ is integrable, i.e. $ \hat{k}\in L^1$ then we can relax the $C^2$ condition on $\tau$ used in Theorem \ref{thm:main_C2} to $\tau \in W^{2,1}$ and the same conclusions as in Theorem \ref{thm:main_C2} hold. 

The reference term of the quasi-local energy \eqref{eq:QLE_trouble_term} is bounded by the assumption that $\hat{k}\in L^1$. Applying the elliptic estimate \cite{martinazzi2009ellipticestimate}
$$\int_\Sigma |\nabla\tau|\leq C \int_\Sigma |\Delta\tau|$$
to the physical term, one has
\begin{align*}
    &\int_\Sigma \sqrt{(\Delta\tau)^2 + |H|^2(1+|\nabla\tau|^2)} + \nabla\tau \cdot \nabla\ln|\Theta_-| -\frac{1}{2}\langle \overline{\nabla}_{\nabla\tau}\, l^-,l^+\rangle +\tau\Delta y \\
    &\leq C' \int_\Sigma |\Delta\tau| + \int_\Sigma \tau \Delta y
\end{align*}
which still blows down to $-\infty$ for nonconstant $\tau$ by Lemma \ref{lemma:trouble_term}. So the conclusion about the quasi-local energy in Theorem \ref{thm:main_C2} remains true for $\tau \in W^{2,1}$.

Rewrite the solvability condition \eqref{eq:OEE_solvability} for the optimal embedding equation as 
    \begin{align*}
        \int_{\hat{\Sigma}} \hat{k} = &\int_\Sigma \tau \, \Delta y + \frac{|\nabla\tau|^2}{1+|\nabla\tau|^2}\sqrt{(\Delta\tau)^2 + |H|^2(1+|\nabla\tau|^2)}+\nabla\tau\cdot\nabla\ln|\Theta_-| + \frac{1}{2}\langle \overline{\nabla}_{\nabla\tau}l^-,l^+\rangle \\
        &+\int_\Sigma \frac{1}{1+|\nabla\tau|^2}\sqrt{(\Delta\tau)^2 + |H_0|^2(1+|\nabla\tau|^2)}
    \end{align*}
By assumption $\hat{k}\in L^1$, the LHS is bounded while the RHS still blows down to $-\infty$ by a similar argument as above. The conclusion about the optimal embedding equation in Theorem \ref{thm:main_C2} remains true for $\tau \in W^{2,1}$.

We also note that a different method is used to prove the second case of Theorem \ref{thm:main_C2} in an independent study \cite{chen2024near}.
\end{rmk}

\begin{rmk}
    The results here are consistent with the previous extension of Wang-Yau quasi-local mass to an horizon in the numerical study \cite{pookzhao2023properties} through 
    $$QLM = \frac{1}{8\pi} \int_\Sigma |H_0|$$
    where $\tau$ was shown to approach a constant toward the horizon at a particular rate to maintain $\alpha_H+\nabla\theta=0$ ($j\equiv 0$ in this axisymmetirc, vanishing-angular-momentum case).
\end{rmk}

\begin{rmk}
It is known that the outer horizon (and nearby surfaces) of a rapidly rotating Kerr black hole ($a>\sqrt{3}/2 M$) has negative Gaussian curvature near the poles; furthermore, explicit calculations show that the outer horizon (and nearby surfaces) of a rapidly rotating Kerr black hole does not admit an (axisymmetric) isometric embedding into $\R^3$. Then our result implies that for a rapidly rotating Kerr black hole, the Wang-Yau quasi-local energy blows up at the horizon and the optimal embedding equation does not admit a solution for surfaces near the horizon.

In the numerical study \cite{miller2018wang}, the minimum of Wang-Yau quasi-local energy was determined numerically without solving the optimal embedding equation. Their minimal quasi-local energy (Figure 6) for surfaces close to the horizon in an extremal Kerr is consistent with a blowup limit in the Wang-Yau quasi-local energy predicted by results here.
\end{rmk}

\begin{rmk}\label{rmk:Dunajski}
When a surface $\Sigma$ cannot be isometrically embedded into $\R^3$, Dunajski and Tod \cite{dunajski2021kijowski} considered isometric embedding into the hyperbolic space. They defined a modified Liu-Yau energy as
$$E_{mLY} = \frac{1}{8\pi} \int_\Sigma H_\HH - |H|$$
where $H_\HH$ is the mean curvature of the isometric embedding into $\HH^3$ of constant section curvature $-\kappa^2$, with $-\kappa^2$ the greatest lower bound of the non-positive Gaussian curvature $K_\Sigma$. They found a finite value for their quasi-local energy even at the horizon  of an extremal Kerr. We note that the difference between their modified Liu-Yau energy and the Wang-Yau quasi-local energy already manifests in their analysis (compare Figure 3 in \cite{dunajski2021kijowski} with Figure 8 in \cite{miller2018wang}).
\end{rmk}

\section{Conclusion and discussion}\label{sec:concludingrmk}

Since the optimal embedding equation is a fourth order PDE for the time function $\tau$, there is no complete analysis of the existence and uniqueness problem \cite{ChenWangYau:2014cmp}.
However, there are cases where a solution can be easily obtained. Assuming the mean curvature vector is spacelike, i.e. $\langle H, H \rangle >0$, constant $\tau$ is an solution to the optimal embedding equation if and only if 
$$\nabla \cdot \alpha_H =0$$
as $\theta=\sinh^{-1}\frac{-\Delta\tau}{|H|\sqrt{1+|\nabla\tau|^2}}=0$ and $\alpha_{v_0}= \langle \nabla^{\R^{3,1}}v_0, T_0\rangle =0$ when $\tau=0$ and $|H|>0$. 
One can show that coordinate spheres in a Kerr (or Kerr-Newman) black hole satisfy $\nabla \cdot \alpha_H =0$. Let us start the discussion with a Kerr black hole. 

A large coordinate sphere of a Kerr black hole, either in an asymptotically flat Cauchy slice or in a null cone has positive Gaussian curvature and hence can be embedded into $\R^3$ by Weyl's theorem \cite{nirenberg1953weyl,pogorelov1952regularity}. 
For a slowly rotating Kerr with $|a|<\frac{\sqrt{3}}{2}M$, its outer horizon has positive Gaussian curvature. So for a slowly rotating Kerr black hole, a fixed constant $\tau$ function solves the optimal embedding equation for coordinate spheres up to the horizon and yields a finite limit for the Wang-Yau quasi-local mass at the horizon. On the other hand, as noted in Remark \ref{rmk:Dunajski}, the horizon of a rapidly rotating Kerr cannot be isometrically embeded into $\R^3$. Then Theorem \ref{thm:main_C2} and Remark \ref{rmk:chen_W21} implies that the Wang-Yau quasi-local energy blows up as $|H|\to 0^+$ while the optimal embedding equation does not admit a solution near the horizon. This indicates that near the horizon of a rapidly rotating Kerr black hole, the minimum of the Wang-Yau quasi-local energy is not a critical point of the variational problem but merely lies on the boundary separating admissible $\tau$ functions from inadmissible $\tau$ functions. This is what Figure 4 of the numerical study \cite{miller2018wang} shows. It remains a question whether there exists any solution to the optimal embedding equation for a general surface in Kerr which does not admit an isometric embedding into $\R^3$. 
In summary, as the surface under consideration moves from the infinity toward the horizon in a Kerr black hole, the Wang-Yau quasi-local energy may blow up or not, depending on the value of the angular momentum. This is summarized in Figure \ \ref{fig:QLE_tau}.

One can easily generalize the above discussion with a Kerr black hole to an asymptotically flat spacetime.  Bearing in mind that the Wang-Yau quasi-local energy is invariant under the action of the Poincare group on $\R^{3,1}$, we will use $\tau=const$ and $\tau=0$ interchangeably.
For a large sphere near the infinity in an asymptotically flat spacetime, either in an asymptotically flat Cauchy slice or in a null cone, $\tau \sim o(r)$ solves the optimal embedding equation \cite{WangYau:2010spatialinf,Chen:2010tz}). Therefore for large spheres, $\tau=0$ (or any constant) is close to a solution of the optimal embedding equation and yields a local minimum for quasi-local energy.
For surfaces approaching an horizon, Theorem \ref{thm:main_C2} or Remark \ref{rmk:chen_W21} predicts that limits of the Wang-Yau quasi-local energy may not exist depending on whether the horizon can be isometrically embedded into $\R^3$ or not. Therefore, Figure \ref{fig:QLE_tau} qualitatively applies. Our study reveals new aspects of the strong field behavior of Wang-Yau quasi-local energy. 

Recall that the Brown-York energy is not defined for a space-like $2$-surface $\Sigma$ that does not admit an isometric embedding into $\R^3$ and the Wang-Yau quasi-local energy alleviates this problem by considering isometric embedding into $R^{3,1}$. However, the results here suggest that further studies are needed to fully address this problem. It could be possible that one needs a different gauge condition than \eqref{eq:canonical_gauge} or one needs to consider isometric embedding into a hyperbolic space as in \cite{dunajski2021kijowski}. Given that the holography principle is more well-established in asymptotically Anti-de Sitter (AdS) spacetime \cite{maldacena1999AdSCFT} than in asymptotically flat spacetime \cite{donnayStrominger2019conformally,donnay2022carrollian}, one might wonder that studying quasi-local mass in asymptotically AdS spacetime would give some clues. Preliminary attempts along this direction are made by \cite{chen2020dSAdS}. Defining quasi-local mass for more general space-like $2$-surfaces like higher-genus or disconnected surfaces is also of interest. A recent related definition using level set technique could help address this question \cite{AKY2024newQLM}.
Lastly, recall that Kijowski found two different forms of gravitational quasi-local energy corresponding to different boundary data. One of them plays the role of internal energy and is similar to Hawking mass while the other one plays the role of free energy and is similar to  Brown-York mass. Therefore, any possible connection between Wang-Yau (type) quasi-local energy and other definitions, e.g. Hawking mass, deserves further study, especially noting that the Hawking mass would take a finite limit of $\sqrt{\frac{A(\Sigma)}{16\pi}}$ as $\Sigma$ approaches an horizon from outside, where $A(\Sigma)$ is the surface area of $\Sigma$. 

\begin{figure}
    \centering
    \includegraphics[scale=.8,trim={3cm 7cm 0 4cm},clip]{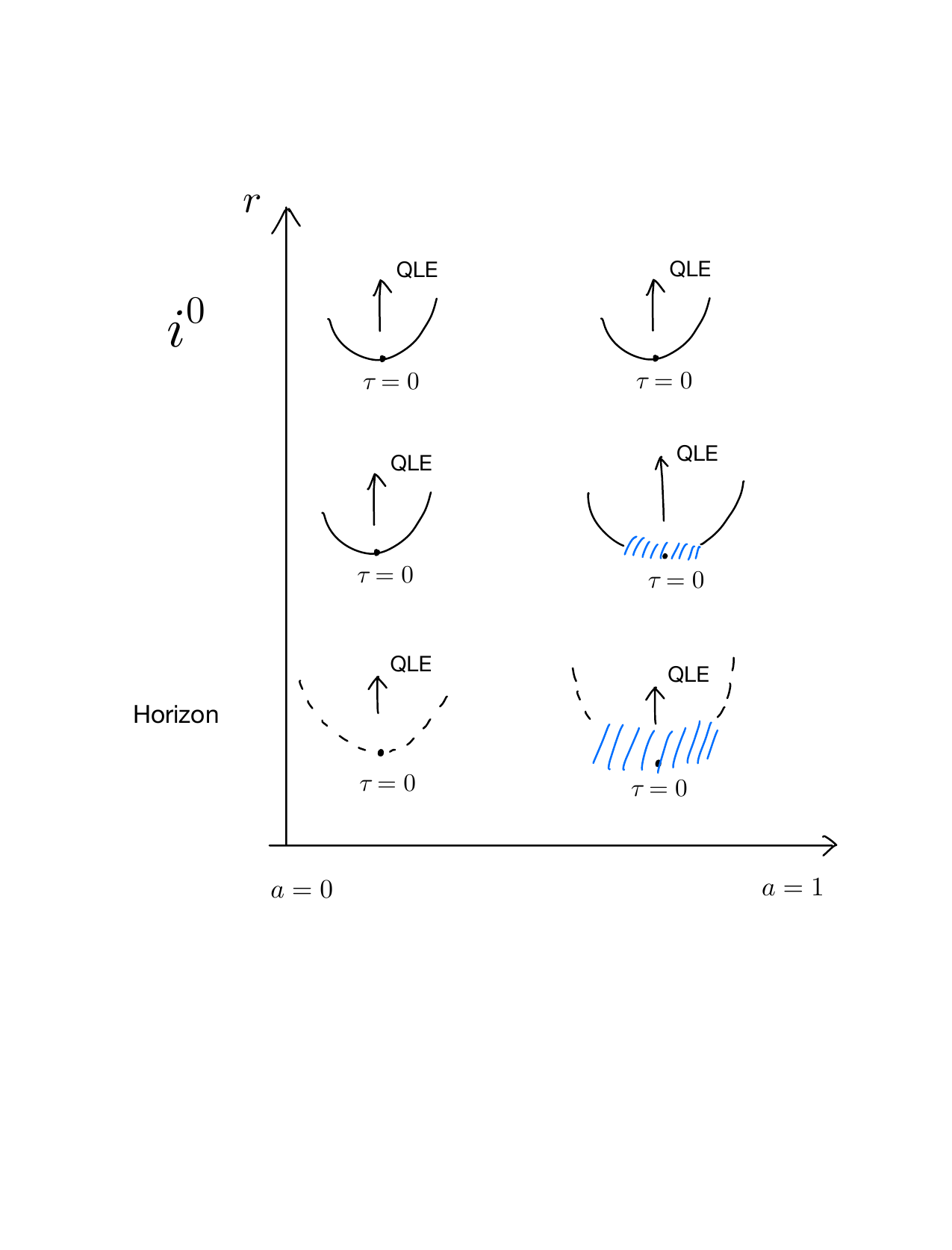}
    \caption{Illustration of Wang-Yau quasi-local energy as a functional of $\tau$-function in the $(r,a)$ space. Here $a$ measures the angular momentum of the black hole spacetime and is related to whether the horizon can be embedded into $\R^3$ while $r$ denotes the distance to the gravitational source and indicates the gravitational field strength. Dashed lines indicate that the Wang-Yau QLE blows up for such $\tau$. Blue shading indicates that isometric embedding into $\R^{3,1}$ with such $\tau$ does not exist and hence not allowed. A phase transition related to the angular momentum is manifest.}
    \label{fig:QLE_tau}
\end{figure}

\bmhead{Data availability}

We do not analyse or generate any datasets, because our work proceeds within a theoretical and mathematical approach. One can obtain the relevant materials from the references below.

\bmhead{Declarations}
The authors have no competing interests to declare that are relevant to the content of this article.

\bmhead{Acknowledgements}

We acknowledge Sergio Cecotti for his inspiring comments and discussions. We also acknowledge a discussion with Mu-Tao Wang during his visit in Beijing. We thank Po-Ning Chen for his suggestion on improving our results.


\bibliography{sn-bibliography}


\begin{thebibliography}{36}
\ifx \bisbn   \undefined \def \bisbn  #1{ISBN #1}\fi
\ifx \binits  \undefined \def \binits#1{#1}\fi
\ifx \bauthor  \undefined \def \bauthor#1{#1}\fi
\ifx \batitle  \undefined \def \batitle#1{#1}\fi
\ifx \bjtitle  \undefined \def \bjtitle#1{#1}\fi
\ifx \bvolume  \undefined \def \bvolume#1{\textbf{#1}}\fi
\ifx \byear  \undefined \def \byear#1{#1}\fi
\ifx \bissue  \undefined \def \bissue#1{#1}\fi
\ifx \bfpage  \undefined \def \bfpage#1{#1}\fi
\ifx \blpage  \undefined \def \blpage #1{#1}\fi
\ifx \burl  \undefined \def \burl#1{\textsf{#1}}\fi
\ifx \doiurl  \undefined \def \doiurl#1{\url{https://doi.org/#1}}\fi
\ifx \betal  \undefined \def \betal{\textit{et al.}}\fi
\ifx \binstitute  \undefined \def \binstitute#1{#1}\fi
\ifx \binstitutionaled  \undefined \def \binstitutionaled#1{#1}\fi
\ifx \bctitle  \undefined \def \bctitle#1{#1}\fi
\ifx \beditor  \undefined \def \beditor#1{#1}\fi
\ifx \bpublisher  \undefined \def \bpublisher#1{#1}\fi
\ifx \bbtitle  \undefined \def \bbtitle#1{#1}\fi
\ifx \bedition  \undefined \def \bedition#1{#1}\fi
\ifx \bseriesno  \undefined \def \bseriesno#1{#1}\fi
\ifx \blocation  \undefined \def \blocation#1{#1}\fi
\ifx \bsertitle  \undefined \def \bsertitle#1{#1}\fi
\ifx \bsnm \undefined \def \bsnm#1{#1}\fi
\ifx \bsuffix \undefined \def \bsuffix#1{#1}\fi
\ifx \bparticle \undefined \def \bparticle#1{#1}\fi
\ifx \barticle \undefined \def \barticle#1{#1}\fi
\bibcommenthead
\ifx \bconfdate \undefined \def \bconfdate #1{#1}\fi
\ifx \botherref \undefined \def \botherref #1{#1}\fi
\ifx \url \undefined \def \url#1{\textsf{#1}}\fi
\ifx \bchapter \undefined \def \bchapter#1{#1}\fi
\ifx \bbook \undefined \def \bbook#1{#1}\fi
\ifx \bcomment \undefined \def \bcomment#1{#1}\fi
\ifx \oauthor \undefined \def \oauthor#1{#1}\fi
\ifx \citeauthoryear \undefined \def \citeauthoryear#1{#1}\fi
\ifx \endbibitem  \undefined \def \endbibitem {}\fi
\ifx \bconflocation  \undefined \def \bconflocation#1{#1}\fi
\ifx \arxivurl  \undefined \def \arxivurl#1{\textsf{#1}}\fi
\csname PreBibitemsHook\endcsname

\bibitem[\protect\citeauthoryear{Penrose}{1982}]{Penrose1982unsolved}
\begin{bbook}
\bauthor{\bsnm{Penrose}, \binits{R.}}:
In: \beditor{\bsnm{Yau}, \binits{S.-T.}} (ed.)
\bbtitle{Some Unsolved Problems in Classical General Relativity},
pp. \bfpage{631}--\blpage{668}.
\bpublisher{Princeton University Press},
\blocation{Princeton}
(\byear{1982}).
\doiurl{10.1515/9781400881918-034} .
\burl{https://doi.org/10.1515/9781400881918-034}
\end{bbook}
\endbibitem

\bibitem[\protect\citeauthoryear{Coley}{2017}]{coley2017open}
\begin{barticle}
\bauthor{\bsnm{Coley}, \binits{A.A.}}:
\batitle{Open problems in mathematical physics}.
\bjtitle{Physica Scripta}
\bvolume{92}(\bissue{9}),
\bfpage{093003}
(\byear{2017})
\end{barticle}
\endbibitem

\bibitem[\protect\citeauthoryear{Hawking}{1968}]{hawking1968hawkingmass}
\begin{barticle}
\bauthor{\bsnm{Hawking}, \binits{S.W.}}:
\batitle{Gravitational radiation in an expanding universe}.
\bjtitle{Journal of Mathematical Physics}
\bvolume{9}(\bissue{4}),
\bfpage{598}--\blpage{604}
(\byear{1968})
\end{barticle}
\endbibitem

\bibitem[\protect\citeauthoryear{Bartnik}{1989}]{bartnik1989newqlm}
\begin{barticle}
\bauthor{\bsnm{Bartnik}, \binits{R.}}:
\batitle{New definition of quasilocal mass}.
\bjtitle{Physical review letters}
\bvolume{62}(\bissue{20}),
\bfpage{2346}
(\byear{1989})
\end{barticle}
\endbibitem

\bibitem[\protect\citeauthoryear{Brown and York~Jr}{1993}]{brown1993quasilocal}
\begin{barticle}
\bauthor{\bsnm{Brown}, \binits{J.D.}},
\bauthor{\bsnm{York~Jr}, \binits{J.W.}}:
\batitle{Quasilocal energy and conserved charges derived from the gravitational action}.
\bjtitle{Physical Review D}
\bvolume{47}(\bissue{4}),
\bfpage{1407}
(\byear{1993})
\end{barticle}
\endbibitem

\bibitem[\protect\citeauthoryear{Liu and Yau}{2003}]{liuYau2003}
\begin{barticle}
\bauthor{\bsnm{Liu}, \binits{C.-C.M.}},
\bauthor{\bsnm{Yau}, \binits{S.-T.}}:
\batitle{Positivity of quasilocal mass}.
\bjtitle{Phys. Rev. Lett.}
\bvolume{90},
\bfpage{231102}
(\byear{2003})
\doiurl{10.1103/PhysRevLett.90.231102}
\end{barticle}
\endbibitem

\bibitem[\protect\citeauthoryear{Liu and Yau}{2006}]{liuYau2006}
\begin{barticle}
\bauthor{\bsnm{Liu}, \binits{C.-C.M.}},
\bauthor{\bsnm{Yau}, \binits{S.-T.}}:
\batitle{Positivity of quasi-local mass {II}}.
\bjtitle{Journal of the American Mathematical Society}
\bvolume{19}(\bissue{1}),
\bfpage{181}--\blpage{204}
(\byear{2006})
\doiurl{10.1090/S0894-0347-05-00497-2}
\end{barticle}
\endbibitem

\bibitem[\protect\citeauthoryear{Szabados}{2009}]{Szabados:2009review}
\begin{barticle}
\bauthor{\bsnm{Szabados}, \binits{L.B.}}:
\batitle{Quasi-local energy-momentum and angular momentum in general relativity}.
\bjtitle{Living reviews in relativity}
\bvolume{12},
\bfpage{1}--\blpage{163}
(\byear{2009})
\doiurl{10.12942/lrr-2009-4}
\end{barticle}
\endbibitem

\bibitem[\protect\citeauthoryear{Wang}{2015}]{wang2015four}
\begin{botherref}
\oauthor{\bsnm{Wang}, \binits{M.-T.}}:
Four lectures on quasi-local mass.
arXiv preprint arXiv:1510.02931
(2015)
\end{botherref}
\endbibitem

\bibitem[\protect\citeauthoryear{Murchadha et~al.}{2004}]{murchadha2004comment}
\begin{barticle}
\bauthor{\bsnm{Murchadha}, \binits{N.O.}},
\bauthor{\bsnm{Szabados}, \binits{L.B.}},
\bauthor{\bsnm{Tod}, \binits{K.P.}}:
\batitle{Comment on ``positivity of quasilocal mass''}.
\bjtitle{Phys. Rev. Lett.}
\bvolume{92},
\bfpage{259001}
(\byear{2004})
\doiurl{10.1103/PhysRevLett.92.259001}
\end{barticle}
\endbibitem

\bibitem[\protect\citeauthoryear{Wang and Yau}{2009a}]{wangyau2009cmp}
\begin{barticle}
\bauthor{\bsnm{Wang}, \binits{M.-T.}},
\bauthor{\bsnm{Yau}, \binits{S.-T.}}:
\batitle{Isometric embeddings into the minkowski space and new quasi-local mass}.
\bjtitle{Communications in Mathematical Physics}
\bvolume{288}(\bissue{3}),
\bfpage{919}--\blpage{942}
(\byear{2009})
\end{barticle}
\endbibitem

\bibitem[\protect\citeauthoryear{Wang and Yau}{2009b}]{wang2009quasilocalPRL}
\begin{barticle}
\bauthor{\bsnm{Wang}, \binits{M.-T.}},
\bauthor{\bsnm{Yau}, \binits{S.-T.}}:
\batitle{Quasilocal mass in general relativity}.
\bjtitle{Physical review letters}
\bvolume{102}(\bissue{2}),
\bfpage{021101}
(\byear{2009})
\end{barticle}
\endbibitem

\bibitem[\protect\citeauthoryear{Wang and Yau}{2010}]{WangYau:2010spatialinf}
\begin{barticle}
\bauthor{\bsnm{Wang}, \binits{M.-T.}},
\bauthor{\bsnm{Yau}, \binits{S.-T.}}:
\batitle{Limit of quasilocal mass at spatial infinity}.
\bjtitle{Commun. Math. Phys.}
\bvolume{296}(\bissue{1}),
\bfpage{271}--\blpage{283}
(\byear{2010})
\doiurl{10.1007/s00220-010-0990-2}
\end{barticle}
\endbibitem

\bibitem[\protect\citeauthoryear{Chen et~al.}{2011}]{Chen:2010tz}
\begin{barticle}
\bauthor{\bsnm{Chen}, \binits{P.-N.}},
\bauthor{\bsnm{Wang}, \binits{M.-T.}},
\bauthor{\bsnm{Yau}, \binits{S.-T.}}:
\batitle{{Evaluating quasilocal energy and solving optimal embedding equation at null infinity}}.
\bjtitle{Commun. Math. Phys.}
\bvolume{308},
\bfpage{845}--\blpage{863}
(\byear{2011})
\doiurl{10.1007/s00220-011-1362-2}
{\href{https://arxiv.org/abs/1002.0927}{{arXiv:1002.0927}}}
{[math.DG]}
\end{barticle}
\endbibitem

\bibitem[\protect\citeauthoryear{Chen et~al.}{2018}]{ChenWangYau:2018smalllimit}
\begin{barticle}
\bauthor{\bsnm{Chen}, \binits{P.-N.}},
\bauthor{\bsnm{Wang}, \binits{M.-T.}},
\bauthor{\bsnm{Yau}, \binits{S.-T.}}:
\batitle{Evaluating small sphere limit of the wang--yau quasi-local energy}.
\bjtitle{Commun. Math. Phys.}
\bvolume{357},
\bfpage{731}--\blpage{774}
(\byear{2018})
\end{barticle}
\endbibitem

\bibitem[\protect\citeauthoryear{Chen}{}]{ChenPN}
\begin{botherref}
\oauthor{\bsnm{Chen}, \binits{P.-N.}}
personal communication
\end{botherref}
\endbibitem

\bibitem[\protect\citeauthoryear{Chen}{2024}]{chen2024near}
\begin{botherref}
\oauthor{\bsnm{Chen}, \binits{P.-N.}}:
Near horizon limit of the wang--yau quasi-local mass.
arXiv preprint arXiv:2408.02917
(2024)
\end{botherref}
\endbibitem

\bibitem[\protect\citeauthoryear{Wang and Yau}{2009}]{wang2009isometric}
\begin{barticle}
\bauthor{\bsnm{Wang}, \binits{M.-T.}},
\bauthor{\bsnm{Yau}, \binits{S.-T.}}:
\batitle{Isometric embeddings into the minkowski space and new quasi-local mass}.
\bjtitle{Communications in Mathematical Physics}
\bvolume{288}(\bissue{3}),
\bfpage{919}--\blpage{942}
(\byear{2009})
\end{barticle}
\endbibitem

\bibitem[\protect\citeauthoryear{Pogorelov}{1964}]{pogorelov1964embedding}
\begin{barticle}
\bauthor{\bsnm{Pogorelov}, \binits{A.}}:
\batitle{Some results on surface theory in the large}.
\bjtitle{Advances in mathematics}
\bvolume{1}(\bissue{2}),
\bfpage{191}--\blpage{264}
(\byear{1964})
\end{barticle}
\endbibitem

\bibitem[\protect\citeauthoryear{Zhao et~al.}{2024}]{zhaoAY2024remarks}
\begin{botherref}
\oauthor{\bsnm{Zhao}, \binits{B.}},
\oauthor{\bsnm{Andersson}, \binits{L.}},
\oauthor{\bsnm{Yau}, \binits{S.-T.}}:
Some remarks on wang-yau quasi-local mass.
arXiv preprint arXiv:2402.19310
(2024)
\end{botherref}
\endbibitem

\bibitem[\protect\citeauthoryear{Chen et~al.}{2018}]{chen2018small}
\begin{barticle}
\bauthor{\bsnm{Chen}, \binits{P.-N.}},
\bauthor{\bsnm{Wang}, \binits{M.-T.}},
\bauthor{\bsnm{Yau}, \binits{S.-T.}}:
\batitle{Evaluating small sphere limit of the wang--yau quasi-local energy}.
\bjtitle{Communications in Mathematical Physics}
\bvolume{357},
\bfpage{731}--\blpage{774}
(\byear{2018})
\end{barticle}
\endbibitem

\bibitem[\protect\citeauthoryear{Chen et~al.}{2014}]{ChenWangYau:2014cmp}
\begin{barticle}
\bauthor{\bsnm{Chen}, \binits{P.-N.}},
\bauthor{\bsnm{Wang}, \binits{M.-T.}},
\bauthor{\bsnm{Yau}, \binits{S.-T.}}:
\batitle{Minimizing properties of critical points of quasi-local energy}.
\bjtitle{Commun. Math. Phys.}
\bvolume{329},
\bfpage{919}--\blpage{935}
(\byear{2014})
\doiurl{10.1007/s00220-014-1909-0}
\end{barticle}
\endbibitem

\bibitem[\protect\citeauthoryear{Andersson et~al.}{2011}]{andersson2011jangreview}
\begin{barticle}
\bauthor{\bsnm{Andersson}, \binits{L.}},
\bauthor{\bsnm{Eichmair}, \binits{M.}},
\bauthor{\bsnm{Metzger}, \binits{J.}}:
\batitle{Jang’s equation and its applications to marginally trapped surfaces}.
\bjtitle{Complex Analysis and Dynamical Systems IV: Part}
\bvolume{2},
\bfpage{13}--\blpage{46}
(\byear{2011})
\end{barticle}
\endbibitem

\bibitem[\protect\citeauthoryear{Zhao}{2021}]{zhaoKW2021blowupJang}
\begin{botherref}
\oauthor{\bsnm{Zhao}, \binits{K.-W.}}:
On blowup of regularized solutions to jang equation and constant expansion surfaces.
arXiv preprint arXiv:2112.02490
(2021)
\end{botherref}
\endbibitem

\bibitem[\protect\citeauthoryear{Chen}{2011}]{Chenthesis}
\begin{botherref}
\oauthor{\bsnm{Chen}, \binits{P.-N.}}:
Quasi-local mass in general relativity.
Ph. d. thesis,
Harvard University
(June 2011)
\end{botherref}
\endbibitem

\bibitem[\protect\citeauthoryear{Martinazzi}{2009}]{martinazzi2009ellipticestimate}
\begin{barticle}
\bauthor{\bsnm{Martinazzi}, \binits{L.}}:
\batitle{Concentration--compactness phenomena in the higher order liouville's equation}.
\bjtitle{Journal of Functional Analysis}
\bvolume{256}(\bissue{11}),
\bfpage{3743}--\blpage{3771}
(\byear{2009})
\end{barticle}
\endbibitem

\bibitem[\protect\citeauthoryear{Pook-Kolb et~al.}{2023}]{pookzhao2023properties}
\begin{barticle}
\bauthor{\bsnm{Pook-Kolb}, \binits{D.}},
\bauthor{\bsnm{Zhao}, \binits{B.}},
\bauthor{\bsnm{Andersson}, \binits{L.}},
\bauthor{\bsnm{Krishnan}, \binits{B.}},
\bauthor{\bsnm{Yau}, \binits{S.-T.}}:
\batitle{Properties of quasilocal mass in binary black hole mergers}.
\bjtitle{Physical Review D}
\bvolume{108}(\bissue{12}),
\bfpage{124031}
(\byear{2023})
\end{barticle}
\endbibitem

\bibitem[\protect\citeauthoryear{Miller et~al.}{2018}]{miller2018wang}
\begin{barticle}
\bauthor{\bsnm{Miller}, \binits{W.A.}},
\bauthor{\bsnm{Ray}, \binits{S.}},
\bauthor{\bsnm{Wang}, \binits{M.-T.}},
\bauthor{\bsnm{Yau}, \binits{S.-T.}}:
\batitle{Wang and yau’s quasi-local energy for an extreme kerr spacetime}.
\bjtitle{Classical and Quantum Gravity}
\bvolume{35}(\bissue{5}),
\bfpage{055007}
(\byear{2018})
\end{barticle}
\endbibitem

\bibitem[\protect\citeauthoryear{Dunajski and Tod}{2021}]{dunajski2021kijowski}
\begin{barticle}
\bauthor{\bsnm{Dunajski}, \binits{M.}},
\bauthor{\bsnm{Tod}, \binits{P.}}:
\batitle{The kijowski--liu--yau quasi-local mass of the kerr black hole horizon}.
\bjtitle{Classical and Quantum Gravity}
\bvolume{38}(\bissue{23}),
\bfpage{235001}
(\byear{2021})
\end{barticle}
\endbibitem

\bibitem[\protect\citeauthoryear{Nirenberg}{1953}]{nirenberg1953weyl}
\begin{barticle}
\bauthor{\bsnm{Nirenberg}, \binits{L.}}:
\batitle{The weyl and minkowski problems in differential geometry in the large}.
\bjtitle{Communications on pure and applied mathematics}
\bvolume{6}(\bissue{3}),
\bfpage{337}--\blpage{394}
(\byear{1953})
\end{barticle}
\endbibitem

\bibitem[\protect\citeauthoryear{Pogorelov}{1952}]{pogorelov1952regularity}
\begin{barticle}
\bauthor{\bsnm{Pogorelov}, \binits{A.V.}}:
\batitle{Regularity of a convex surface with given gaussian curvature}.
\bjtitle{Matematicheskii Sbornik}
\bvolume{73}(\bissue{1}),
\bfpage{88}--\blpage{103}
(\byear{1952})
\end{barticle}
\endbibitem

\bibitem[\protect\citeauthoryear{Maldacena}{1999}]{maldacena1999AdSCFT}
\begin{barticle}
\bauthor{\bsnm{Maldacena}, \binits{J.}}:
\batitle{The large-n limit of superconformal field theories and supergravity}.
\bjtitle{International journal of theoretical physics}
\bvolume{38}(\bissue{4}),
\bfpage{1113}--\blpage{1133}
(\byear{1999})
\end{barticle}
\endbibitem

\bibitem[\protect\citeauthoryear{Donnay et~al.}{2019}]{donnayStrominger2019conformally}
\begin{barticle}
\bauthor{\bsnm{Donnay}, \binits{L.}},
\bauthor{\bsnm{Puhm}, \binits{A.}},
\bauthor{\bsnm{Strominger}, \binits{A.}}:
\batitle{Conformally soft photons and gravitons}.
\bjtitle{Journal of High Energy Physics}
\bvolume{2019}(\bissue{1}),
\bfpage{1}--\blpage{22}
(\byear{2019})
\end{barticle}
\endbibitem

\bibitem[\protect\citeauthoryear{Donnay et~al.}{2022}]{donnay2022carrollian}
\begin{barticle}
\bauthor{\bsnm{Donnay}, \binits{L.}},
\bauthor{\bsnm{Fiorucci}, \binits{A.}},
\bauthor{\bsnm{Herfray}, \binits{Y.}},
\bauthor{\bsnm{Ruzziconi}, \binits{R.}}:
\batitle{Carrollian perspective on celestial holography}.
\bjtitle{Physical Review Letters}
\bvolume{129}(\bissue{7}),
\bfpage{071602}
(\byear{2022})
\end{barticle}
\endbibitem

\bibitem[\protect\citeauthoryear{Chen et~al.}{2020}]{chen2020dSAdS}
\begin{barticle}
\bauthor{\bsnm{Chen}, \binits{P.-N.}},
\bauthor{\bsnm{Wang}, \binits{M.-T.}},
\bauthor{\bsnm{Yau}, \binits{S.-T.}}:
\batitle{Quasi-local energy with respect to de sitter/anti-de sitter reference}.
\bjtitle{Communications in Analysis and Geometry}
\bvolume{28}(\bissue{7}),
\bfpage{1489}--\blpage{1531}
(\byear{2020})
\end{barticle}
\endbibitem

\bibitem[\protect\citeauthoryear{Alaee et~al.}{2024}]{AKY2024newQLM}
\begin{barticle}
\bauthor{\bsnm{Alaee}, \binits{A.}},
\bauthor{\bsnm{Khuri}, \binits{M.}},
\bauthor{\bsnm{Yau}, \binits{S.-T.}}:
\batitle{A quasi-local mass}.
\bjtitle{Communications in Mathematical Physics}
\bvolume{405}(\bissue{5}),
\bfpage{111}
(\byear{2024})
\end{barticle}
\endbibitem

\end{thebibliography}

\end{document}